\documentclass[conference]{IEEEtran}

\IEEEoverridecommandlockouts
\usepackage{cite}
\usepackage{bm}

\usepackage{amsmath,amssymb,amsfonts}
\usepackage{algorithmic}
\usepackage{graphicx}
\usepackage{textcomp}
\usepackage{xcolor}

\usepackage{algorithm}
\usepackage{array}
\usepackage{arydshln}
\usepackage{amsfonts} 
\usepackage{amssymb}
\usepackage{esint} 
\usepackage{units}
\usepackage{multirow}
\usepackage{comment}

\usepackage{color}

\usepackage{tikz} 
\usepackage[utf8]{inputenc}
\usepackage{pgfplots} 
\usepackage{pgfgantt}
\usepackage{pdflscape}

\usepackage{acro}

\DeclareAcronym{5g}{
short=5G,
long= fifth generation,
}
\DeclareAcronym{6g}{
short=6G,
long= sixth generation,
}
\DeclareAcronym{3d}{
short=3D,
long= three-dimensional,
}
\DeclareAcronym{2d}{
short=2D,
long= two-dimensional,
}
\DeclareAcronym{aod}{
short=AOD,
long= angle-of-departure,
}
\DeclareAcronym{aosa}{
short=AOSA,
long= array-of-subarray,
}
\DeclareAcronym{adod}{
short=ADOD,
long= angle-difference-of-departure,
}
\DeclareAcronym{aoa}{
short=AOA,
long= angle-of-arrival,
}
\DeclareAcronym{adc}{
short=ADC,
long= analog to digital converter,
}
\DeclareAcronym{aeb}{
short=AEB,
long= angle error bound,
}
\DeclareAcronym{av}{
short=AV,
long= autonomous vehicle,
}
\DeclareAcronym{bs}{
short=BS,
long= base station,
}

\DeclareAcronym{csi}{
short=CSI,
long= channel state information,
}
\DeclareAcronym{cfo}{
short=CFO,
long= carrier frequency offset,
}
\DeclareAcronym{ceb}{
short=CEB,
long= clock error bound,
}
\DeclareAcronym{coa}{
short=COA,
long= curvature-of-arrival,
}
\DeclareAcronym{crb}{
short=CRB,
long= Cram\'er-Rao bound,
}
\DeclareAcronym{ccrb}{
short=CCRB,
long= constrained Cram\'er-Rao bound,
}
\DeclareAcronym{cmos}{
short=CMOS,
long= complementary metal-oxide-semiconductor,
}
\DeclareAcronym{crlb}{
short=CRLB,
long= Cram\'er-Rao lower bound,
}
\DeclareAcronym{cdf}{
short=CDF,
long= cumulative distribution function,
}
\DeclareAcronym{cp}{
short=CP,
long= cyclic prefix,
}
\DeclareAcronym{dac}{
short=DAC,
long= digital to analog converter,
}
\DeclareAcronym{dfl}{
short=DFL,
long= device-free localization,
}
\DeclareAcronym{dmimo}{
short=D-MIMO,
long= distributed MIMO,
}
\DeclareAcronym{dlprs}{
short=DL-PRS,
long= downlink positioning reference signal,
}
\DeclareAcronym{d2d}{
short=D2D,
long= device-to-device,
}
\DeclareAcronym{dftsofdm}{
short=DFT-s-OFDM,
long= discrete-Fourier-transform spread OFDM,
}
\DeclareAcronym{dl}{
short=DL,
long= deep learning,
}
\DeclareAcronym{efim}{
short=EFIM,
long = equivalent Fisher information matrix,
}
\DeclareAcronym{gps}{
short=GPS,
long= global positioning system,
}
\DeclareAcronym{fim}{
short=FIM,
long = Fisher information matrix,
}
\DeclareAcronym{hwi}{
short=HWI,
long= hardware impairment,
}
\DeclareAcronym{hemt}{
short=HEMT,
long= high electron mobility transistor,
}
\DeclareAcronym{hbt}{
short=HBT,
long= heterojunction bipolar transistors,
}
\DeclareAcronym{iot}{
short=IoT,
long= internet of things,
}
\DeclareAcronym{imu}{
short=IMU,
long= inertial measurement unit,
}
\DeclareAcronym{isac}{
short=ISAC,
long= integrated sensing and communication,
}
\DeclareAcronym{iqi}{
short=IQI,
long= in-phase and quadrature imbalance,
}
\DeclareAcronym{ip}{
short=IP,
long= incidence point,
}
\DeclareAcronym{ia}{
short=IA,
long= initial access,
}
\DeclareAcronym{kpi}{
short=KPI,
long= key performance indicator,
}
\DeclareAcronym{kf}{
short=KF,
long= Kalman filter,
}
\DeclareAcronym{ekf}{
short=EKF,
long= extended Kalman filter,
}
\DeclareAcronym{ukf}{
short=UKF,
long= unscented Kalman filter,
}
\DeclareAcronym{ckf}{
short=CKF,
long= cubature Kalman filter,
}
\DeclareAcronym{pf}{
short=PF,
long= particle filter,
}
\DeclareAcronym{lb}{
short=LB,
long= lower bound,
}
\DeclareAcronym{lse}{
short=LSE,
long= least-square estimator,
}
\DeclareAcronym{lo}{
short=LO,
long= local oscillator,
}
\DeclareAcronym{los}{
short=LOS,
long= line-of-sight,
}
\DeclareAcronym{mc}{
short=MC,
long= mutual coupling,
}
\DeclareAcronym{mac}{
short=MAC,
long= medium access control,
}
\DeclareAcronym{meb}{
short=MEB,
long= mapping error bound,
}
\DeclareAcronym{ml}{
short=ML,
long= machine learning,
}
\DeclareAcronym{mcrb}{
short=MCRB,
long= misspecified Cram\'er-Rao bound,
}
\DeclareAcronym{mds}{
short=MDS,
long= multidimensional scaling ,
}
\DeclareAcronym{mimo}{
short=MIMO,
long= multiple-input-multiple-output,
}
\DeclareAcronym{siso}{
short=SISO,
long= single-input-single-output,
}
\DeclareAcronym{mm}{
short=MM,
long= mismatched model,
}
\DeclareAcronym{mpc}{
short=MPC,
long= multipath components,
}
\DeclareAcronym{mmwave}{
short=mmWave,
long= millimeter wave,
}
\DeclareAcronym{mmle}{
short=MMLE,
long= mismatched maximum likelihood estimation,
}
\DeclareAcronym{mems}{
short=MEMS,
long= micro-electro-mechanical system,
}
\DeclareAcronym{mle}{
short=MLE,
long= maximum likelihood estimation,
}
\DeclareAcronym{nlos}{
short=NLOS,
long= none-line-of-sight,
}
\DeclareAcronym{ofdm}{
short=OFDM,
long= orthogonal frequency-division multiplexing,
}
\DeclareAcronym{oeb}{
short=OEB,
long= orientation error bound,
}
\DeclareAcronym{otfs}{
short=OTFS,
long= orthogonal time-frequency space,
}
\DeclareAcronym{pdf}{
short=PDF,
long= probability density function,
}
\DeclareAcronym{papr}{
short=PAPR,
long= peak-to-average-power ratio,
}
\DeclareAcronym{pan}{
short=PAN,
long= power amplifier nonlinearity,
}
\DeclareAcronym{pa}{
short=PA,
long= power amplifier,
}
\DeclareAcronym{ps}{
short=PS,
long= phase shifter,
}
\DeclareAcronym{pn}{
short=PN,
long= phase noise,
}
\DeclareAcronym{poa}{
short=POA,
long= phase-of-arrival,
}
\DeclareAcronym{pwm}{
short=PWM,
long= planar wave model,
}
\DeclareAcronym{pdoa}{
short=PDOA,
long= phase-difference-of-arrival,
}
\DeclareAcronym{prs}{
short=PRS,
long= positioning reference signals,
}
\DeclareAcronym{peb}{
short=PEB,
long= position error bound,
}
\DeclareAcronym{rnn}{
short=RNN,
long= recurrent neural network,
}
\DeclareAcronym{rl}{
short=RL,
long= reinforcement learning,
}
\DeclareAcronym{rfc}{
short=RFC,
long= radio-frequency chain,
}
\DeclareAcronym{rf}{
short=RF,
long= radio frequency,
}
\DeclareAcronym{rfid}{
short=RFID,
long= radio frequency identification,
}
\DeclareAcronym{ris}{
short=RIS,
long= reconfigurable intelligent surface,
}
\DeclareAcronym{rss}{
short=RSS,
long= received signal strength,
}
\DeclareAcronym{rtt}{
short=RTT,
long= round-trip time,
}
\DeclareAcronym{sm}{
short=SM,
long= standard model,
}
\DeclareAcronym{sige}{
short=SiGe,
long= silicon-germanium,
}
\DeclareAcronym{spp}{
short=SPP,
long= surface plasmon polariton,
}
\DeclareAcronym{sa}{
short=SA,
long= subarray,
}
\DeclareAcronym{sota}{
short=SOTA,
long= state-of-the-art,
}
\DeclareAcronym{swm}{
short=SWM,
long= spherical wave model,
}
\DeclareAcronym{slam}{
short=SLAM,
long= simultaneous localization and mapping,
}
\DeclareAcronym{tm}{
short=TM,
long= true model,
}
\DeclareAcronym{toa}{
short=TOA,
long= time-of-arrival,
}
\DeclareAcronym{tof}{
short=TOF,
long= time-of-flight,
}
\DeclareAcronym{tdoa}{
short=TDOA,
long= time-difference-of-arrival,
}
\DeclareAcronym{thz}{
short=THz,
long= terahertz,
}
\DeclareAcronym{ue}{
short=UE,
long= user equipment,
}
\DeclareAcronym{ummimo}{
short=UM-MIMO,
long= ultra-massive multi-input-multi-output,
}
\DeclareAcronym{vlp}{
short=VLP,
long= visible light positioning,
}
\DeclareAcronym{veb}{
short=VEB,
long= velocity error bound,
}
\DeclareAcronym{vlc}{
short=VLC,
long= visible light communication,
}
\DeclareAcronym{ula}{
short=ULA,
long= uniform linear array,
}
\DeclareAcronym{upa}{
short=UPA,
long= uniform planar array,
}
\DeclareAcronym{wlan}{
short=WLAN,
long= wireless local area network,
}

\usepackage{tabu,longtable}
\usepackage{balance}

\usepackage[english]{babel}
\usepackage{amsthm}
\newtheorem{theorem}{Theorem}
\newtheorem{lemma}[theorem]{Lemma}
\newtheorem{prop}{Proposition}

\ifCLASSOPTIONcompsoc
 \usepackage[caption=false,font=normalsize,labelfont=sf,textfont=sf]{subfig}
\else
 \usepackage[caption=false,font=footnotesize]{subfig}
\fi
\usepackage{stfloats}
\usepackage[hidelinks]{hyperref}
\usepackage{xcolor}

\long\def\comment#1{}

\DeclareMathOperator*{\argmin}{arg\,min}

\newfont{\bbb}{msbm10 scaled 700}

\newcommand{\hthickline}{\noalign{\hrule height 0.80pt}}

\newfont{\bb}{msbm10 scaled 1100}

\newcommand{\av}{{\bf a}}
\newcommand{\bv}{{\bf b}}

\newcommand{\hv}{{\bf h}}

\newcommand{\pv}{{\bf p}}

\newcommand{\sv}{{\bf s}}
\newcommand{\tv}{{\bf t}}

\newcommand{\wv}{{\bf w}}
\newcommand{\vv}{{\bf v}}

\newcommand{\Am}{{\bf A}}
\newcommand{\Bm}{{\bf B}}

\newcommand{\Dm}{{\bf D}}
\newcommand{\Em}{{\bf E}}
\newcommand{\Fm}{{\bf F}}

\newcommand{\Jm}{{\bf J}}

\newcommand{\Om}{{\bf O}}

\newcommand{\Qm}{{\bf Q}}

\newcommand{\Xm}{{\bf X}}

\newcommand{\alphav}{\hbox{\boldmath$\alpha$}}

\newcommand{\gammav}{\hbox{\boldmath$\gamma$}}

\newcommand{\xiv}{\hbox{\boldmath$\xi$}}

\newcommand{\trace}{{\hbox{tr}}}

\renewcommand{\arg}{{\hbox{arg}}}

\begin{document}




\title{Doppler-Enabled Single-Antenna Localization and Mapping Without Synchronization}



\author{
\IEEEauthorblockN{Hui Chen, Fan Jiang, Yu Ge, Hyowon Kim, Henk Wymeersch}
\IEEEauthorblockA{
Department of Electrical Engineering, Chalmers University of Technology, Sweden\\
Email: \{hui.chen; fan.jiang; yuge; hyowon; henkw\}@chalmers.se}
}

\maketitle

\begin{abstract}
Radio localization is a key enabler for joint communication and sensing in the fifth/sixth generation (5G/6G) communication systems. With the help of multipath components (MPCs), localization and mapping tasks can be done with a single base station (BS) and single unsynchronized user equipment (UE) if both of them are equipped with an antenna array. However, the antenna array at the UE side increases the hardware and computational cost, preventing localization functionality. In this work, we show that with Doppler estimation and MPCs, localization and mapping tasks can be performed even with a single-antenna mobile UE. Furthermore, we show that the localization and mapping performance will improve and then saturate at a certain level with an increased UE speed. Both theoretical Cram\'er-Rao bound analysis and simulation results show the potential of localization under mobility and the effectiveness of the proposed localization algorithm.
\end{abstract}

\section{Introduction}
Location and map information can assist communication in the
the millimeter wave/Terahertz (mmWave/THz) band 5G/6G systems, as well as a variety of location-based applications such as autonomous driving~\cite{wymeersch20175g}, tactile robots~\cite{haddadin2018tactile}, etc. Due to the geometrical channel property of high frequency radio signals~\cite{shahmansoori2017position}, the localization and mapping tasks can be done by exploiting the channel parameters (e.g., angles and delay) of each path from the estimated channel matrix~\cite{jiang2021beamspace}.
With the help of one or more reference points (e.g., \ac{bs} or \ac{ris}), the position and orientation of the \ac{ue}, as well as the position of the \ac{ip}, can be estimated~\cite{chen2021tutorial}.

Localization and mapping for \ac{ofdm}-based communication systems have been studied extensively; with the assistance of \ac{nlos} paths in \ac{2d} or \ac{3d} scenarios~\cite{shahmansoori2017position, abu2018error}. 
Recent research considers single-snapshot localization and mapping with a single-antenna receiver by assuming the system is synchronized~\cite{fascista2021downlink}.
However, only single snapshot localization\footnote{Even though multiple measurements are performed, they are all within a channel coherence interval so that the channel is assumed to be fixed, and we refer to this as single snapshot localization.} is discussed in these works, and no mobility is involved. 
Nevertheless, UE mobility could be found in a lot of scenarios, such as autonomous driving~\cite{wymeersch20175g} and high speed trains~\cite{talvitie2019positioning}, which degrades the localization accuracy without considering the Doppler effect~\cite{win2018theoretical}. 

The traditional way of dealing with this issue is to reduce the coherence time, within which the channel could be considered fixed. In addition, with a properly designed signal, such as Zadoff-Chu sequence~\cite{kim2022preamble}, the Doppler effect can be mitigated. However, estimating the Doppler could be a better option that can benefit localization and tracking~\cite{amar2008localization, shames2013doppler, han2015performance, kakkavas2019performance}. In~\cite{shames2013doppler}, passive localization using Doppler shift is discussed, proving that a unique position can be estimated with at least 5 Doppler measurements in 2D scenarios. The work in~\cite{han2015performance} shows that the Doppler effect increases the AOA information, which can be interpreted as the enlargement of the virtual array aperture brought by the movement. More recent work considers a single-BS \ac{mimo} \ac{ofdm} mmWave system, showing the NLOS-only scenario can be significantly improved by mobility, depending on the variance of the synchronization error~\cite{kakkavas2019performance}.



In this work, we consider a more challenging MISO scenario to locate a single-antenna \ac{ue} with a single \ac{ula} BS under unknown clock offset. The contributions of this work is summarized as follows: (i)~We show that with a sufficient number of \acp{mpc}, Doppler estimation can enable localization and mapping (which is previously impossible for a stationary UE) by providing extra radial velocity measurements, in addition to angles and delays; (ii)~We prove that the localization and mapping performance improves and then saturates with an increased velocity, while the velocity estimation performance keeps decreasing, which is determined by the geometrical relationship of different paths; and (iii)~We propose a simple 1-D search localization and mapping algorithm, and use simulations and \ac{crb} analysis to show the effectiveness of the algorithm.


\begin{figure}[t]
  \centering
\centerline{\includegraphics[width=0.9\linewidth]{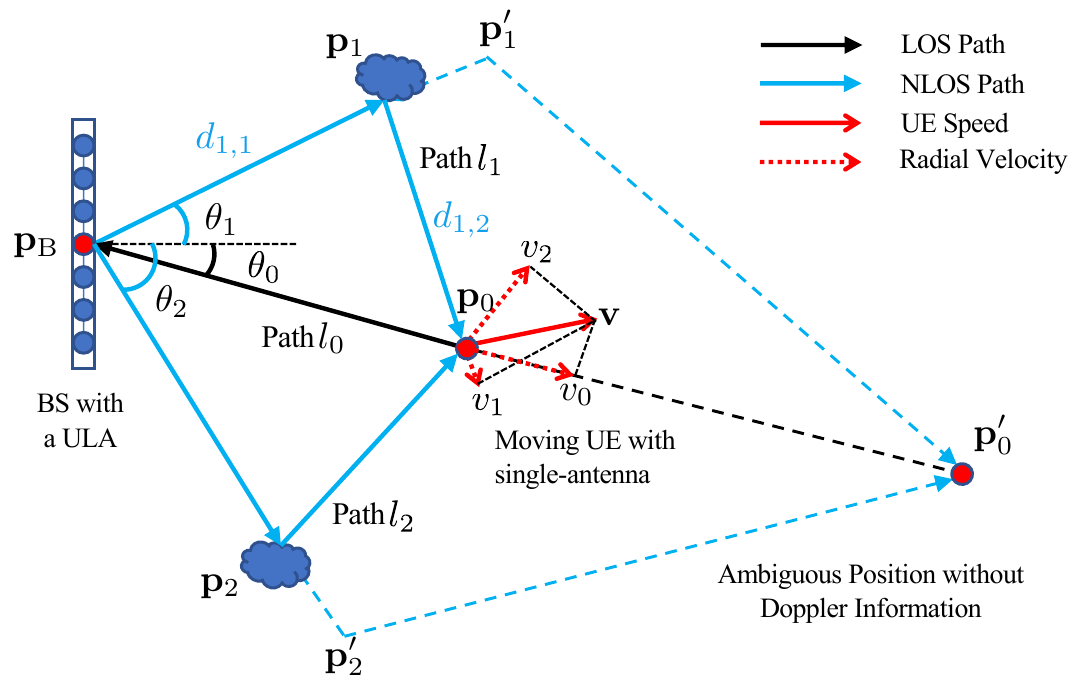}}
\caption{Illustration of Doppler-assisted localization. If the UE is stationary, the localization problem is not solvable where $\pv_0'$ is one of ambiguous positions that produces the same geometry parameters as the true position $\pv_0$ (i.e., AODs and time-difference-of-arrivals).} \vspace{-5mm}
\label{fig-1-illustration}
\end{figure}

\section{Problem Statement}
In this section, we start with the system model and then describe the relationship between the channel parameters and unknown states. How Doppler can assist localization with a sufficient number of \ac{nlos} paths are also discussed.
\subsection{System Model}
We consider a 2D downlink scenario with one BS equipped with a $N_\text{B}$-antenna \ac{ula} (with the center $\pv_\text{B}$ located at the origin of the coordinate system) and
one UE with a single antenna \cite{fascista2019millimeter}. The location of the UE at the $g$th transmission is $\mathbf{p}_{0,g}=\mathbf{p}_{0}+\mathbf{v}gT_\text{int}$ where $T_\text{int}$ is the measurement interval and $\vv$ is the velocity, assumed to be fixed within the observation duration $GT_\text{int}$ (where $G$ is the number of total transmissions). We assume there exist one \ac{los} path ($l=0$) and
$L\ge1$ \ac{nlos} paths (with \acp{ip} located at $\mathbf{p}_{l}, 1\le l \le L$), as illustrated in Fig.~\ref{fig-1-illustration}. If only one \ac{rfc} is adopted at both the \ac{bs} and \ac{ue} side, the observation model (the extension to 3D scenarios are also valid) can be formulated as
\begin{align}
    y_{g,k}=& \mu_{g,k} + n_{g,k} = \hv_{g,k}^\top \wv_{\text{B},g} x_{g,k} + n_{g,k},
    \label{eq:channel_model}
\end{align}
where $\mu_{g,k}$ is the noise-free version of the received signal, $\wv_{\text{B},g}$ is the precoding matrix at the BS of the $g$th transmission containing the phase-shifter coefficients with unit amplitude $|\wv_{\text{B},i}|=1$, $x_{g,k}$ is the transmitted signal symbol of the $g$th transmission at the $k$th subcarrier with a constant average transmission power $P$ ($|x_{g,k}|^2 = {P}$), $n_{g,k}\sim\mathcal{CN}(0, \sigma_n^2)$ is the additive white Gaussian noise. The channel matrix $\hv_{g,k}$ can be expressed as
\begin{align}
    \hv_{g,k}=&\sum_{l=0}^{L}\rho_l \av_{\text{B}}({\theta}_{\text{B},l})
    e^{-j 2\pi\Delta_{f}k\tau_{l}}e^{j 2\pi g T_\text{int} v_{l}/\lambda},
    \label{eq:channel_matrix}
\end{align}
where $\av_\text{B}({\theta})$ is the steering vector at BS, $\Delta_f$ is the subcarrier spacing, $\lambda$ is the wavelength, and $\rho_l$, $\theta_{\text{B},l}$, $\tau_l$, $v_l$ ($0\le l \le L$) are the channel parameters representing the complex channel gain, \ac{aod}, delay, and radial velocity of the $l$th path, respectively. Next, we describe the relationship between the channel parameters and state parameters.


\subsection{Channel Parameters}
Consider the sparse property of the channel, it is possible to use a limited number of  parameters to represent the channel matrix with a much larger dimension as shown in~\eqref{eq:channel_matrix}. We define a \emph{channel parameter vector} as $\gammav = [\gammav_0^\top, \ldots, \gammav_L^\top]^\top$ ($\gammav_l = [\theta_l, \tau_l, v_l]^\top$), and a \emph{state vector} as $\sv = [\pv_0^\top, \pv_1^\top, \ldots, \pv_L^\top, B, \vv^\top]^\top$. Note that for a complete FIM calculation, the nuisance parameters (i.e., complex channel gain $\rho_l=\alpha_l e^{-j\xi_l}$) should be considered such that $\tilde\gammav = [\gammav^\top, \alphav^\top, \xiv^\top]^\top$ and $\tilde\sv = [\sv^\top, \alphav^\top, \xiv^\top]^\top$ (where $\alphav = [\alpha_0, \ldots, \alpha_L]^\top$ and $\xiv = [\xi_0, \ldots, \xi_L]^\top$ contain the amplitude and phase of the channel gain of the $L+1$ paths) are able to present all the unknowns. In the following, we use $\gammav$ and $\sv$ to indicate the unknowns of interest for convenience. For a stationary scenario, the radial velocity $v_l$ as well as the velocity unknown $\vv$ can be removed. The relationship between channel parameters and the state parameters can be expressed as
\begin{align}
    \theta_l & = 
    \begin{cases} 
    \arctan2(t_{\text{B},0,2}, t_{\text{B},0,1}), \tv_{\text{B}, 0} = \frac{\pv_0 - \pv_\text{B}}{d_0} & l = 0\\
    \arctan2(t_{\text{B},l,2}, t_{\text{B},l,1}), \tv_{\text{B}, l} = \frac{\pv_{l} - \pv_\text{B}}{d_{l,1}} & l > 0
    \end{cases},
    \label{eq:angle_at_BS}
    \\
    \tau_l & = \frac{d_l + B}{c}=
    \begin{cases}
    \frac{{\Vert \pv_0 - \pv_\text{B} \Vert} + B}{c} & l=0\\
    \frac{{\Vert \pv_l - \pv_\text{B} \Vert} + {\Vert \pv_l - \pv_0 \Vert} + B}{c} & l>0
    \end{cases},
    \\
    v_{l} & = \mathbf{v}^{\top}\tv_{\text{U}, l} =
    \begin{cases}
    \mathbf{v}^{\top}\frac{\pv_{\text{B}}-\pv_{\text{U}}}{d_{0}} & l=0\\
    \mathbf{v}^{\top}\frac{\pv_{\text{l}}-\pv_{\text{U}}}{d_{l,2}} & l>0
    \end{cases},
    \label{eq:state_velocity}
\end{align}
where $\alphav$ and $\xiv$ are unknown vectors to be estimated, $d_0 = {\Vert \pv_0 - \pv_\text{B} \Vert}$ is the distance between the BS and UE, 
$d_l = d_{l,1} + d_{l,2} = {\Vert \pv_l - \pv_\text{B} \Vert} + {\Vert \pv_l - \pv_0 \Vert}$ 
is the actual signal propagation distance of the $l$th path, $c$ is the speed of light (in [m/s]), $B$ is the clock offset in [m], $\tv_{\text{B}, l}$, and $\tv_{\text{U}, l}$ are the direction vector at the BS and UE, respectively.
The channel parameters contain the geometry information of each path, which is assumed to be obtained by the channel parameter estimation methods (e.g., MD-ESPRIT~\cite{jiang2021beamspace}). The lower bound of the measurement parameter vector $\gammav$ can be obtained using \ac{crb}, which will be detailed in Section III-A.

\subsection{Doppler-assisted Localization}
\label{sec:doppler_assisted_localization}
From the channel model and the relationship between the channel parameter vector $\gammav$ and state vector $\sv$, we are able to explain how Doppler can assist in localization.
For $L$ resolvable NLOS paths, the numbers of unknown UE states and IP positions are $5$ (2D UE position, 2D UE velocity, and clock offset) and $2L$ (2D IP position for $L$ paths), respectively.
By contrast, the number of channel parameters is $3(L+1)$ (AOD, delay, velocity) with UE mobility, and $2(L+1)$ for stationary UE. 
If the number of channel parameters is larger than the number of unknowns, the \ac{fim} of the state unknowns are of full rank, indicating the localization problem is solvable.\footnote{This is valid for the scenarios described in this work with resolvable paths LOS and NLOS paths. However, when the paths are unresolvable, the rank calculation needs to be modified.}

Obviously, localization and mapping tasks cannot be completed for the MISO scenario with a stationary UE. However, when radial velocity can be estimated, the number of channel parameters could be larger than the number of unknowns with a sufficient number of NLOS paths. A summary of the minimum number of NLOS paths needed for localization is shown in Table~\ref{table_1:nlos_path_summary}.
Note that with a mobile UE, the localization and mapping can still be performed without the LOS channel, which is not impossible even with synchronization~\cite{fascista2021downlink}.
This table also provided the scenarios with known velocity information; the number of unknowns is reduced by two. Note that the LOS path only provides 2 channel parameter measurements with known velocity (due to the radial velocity can be calculated directly with known AOD), which can be easily verified by 2D-LOS (known $\vv$) cannot be localized with $L=0$).

\begin{table}[ht]
\scriptsize
    \centering
    \caption{Summary of Minimum Number of NLOS Paths needed for Doppler-assisted Localization}
    \begin{tabular}{ c|c|c|c } 
    \hline
    \shortstack{Localization\\ Scenarios} & \shortstack{Unknown\\ States} & \shortstack{Channel\\ Parameters} & \shortstack{Min. \# of\\ NLOS Paths} \\ 
    \hline
    With LOS (stationary) & $3+2L$ & $2+2L$ & \text{Unsolvable}
    \\
    \hline
    Without LOS (stationary) & $3+2L$ & $2L$ & \text{Unsolvable}
    \\
    \hline
    With LOS (mobile) & $5+2L$ & $3+3L$ & 2 
    \\
    \hline
    Without LOS (mobile) & $5+2L$ & $3L$ & 5 
    \\
    \hline
    With LOS (known $\vv$) & $3+2L$ & $2+3L$ & 1
    \\ 
    \hline
    Without LOS (known $\vv$) & $3+2L$ & $3L$ & 3\\
    \hline
    \end{tabular}
    \vspace{-0.2cm}
    \label{table_1:nlos_path_summary}
\end{table}

\section{Performance Analysis and Localization Algorithm}
\label{sec:performance_analysis_and_localization_algorithm}
In this section, we briefly describe the \ac{crb} of the unknown states (i.e., \ac{peb}, \ac{meb}, \ac{ceb} and \ac{veb}) UE state by deriving the Fisher information matrix under UE mobility. In addition, \ac{fim} analysis is performed to discuss the effect of speed on the estimation of unknowns. A simple localization and mapping algorithm is also proposed with the estimated channel parameters.\vspace{0mm}
\subsection{Cram\'er-Rao Bound}
Based on the channel model defined in~\eqref{eq:channel_model}, the CRB of the state parameters can be obtained as $\mathrm{CRB} \triangleq \left[\mathbf{I}(\sv)\right]^{-1} = \left[\Jm_\mathrm{S} \mathbf{I}({\gammav}) \Jm_\mathrm{S}^\top\right]^{-1}$, 
where $\mathbf{I}(\sv)$ is the \ac{efim} \cite{mendrzik2018harnessing} of the unknown state parameters $\sv$, $\mathbf{I}({\gammav})$ is the \ac{efim} of unknowns of interests $\gammav$ from $\mathbf{I}({\tilde\gammav})$ with
\begin{align}
    \mathbf{I}({\tilde\gammav}) & 
    = \frac{2}{\sigma_n^2}\sum^{G}_{g=1} \sum^K_{k=1}\mathrm{Re}\left\{
    \left(\frac{\partial\mu_{g,k}}{\partial{\tilde\gammav}}\right)^{\mathsf{H}} 
    \left(\frac{\partial\mu_{g,k}}{\partial{\tilde\gammav}}\right)\right\}.
    \label{eq:FIM_measurement}
\end{align}
Here, $\mathrm{Re}\{\cdot\}$ is getting the real part of a complex number, and $\Jm_\mathrm{S}\in \mathbb{R}^{(2L+5)\times(3L)}$ is the Jacobian matrix using a denominator-layout notation from the channel parameter vector $\gammav$ to the state vector $\sv$ as
\begin{equation}
    \Jm_\mathrm{S} \triangleq \frac{\partial {\gammav}}{\partial \sv} = 
    \begin{bmatrix}
        \frac{\partial \gammav_0}{\partial \sv}, & \ldots, & \frac{\partial \gammav_L}{\partial \sv}
    \end{bmatrix}.
    \label{eq:jacobian_state}
\end{equation}
For stationary scenarios, matrices $\mathbf{I}(\gammav)$ and $\Jm_\text{S}$ can be obtained similarly as~\cite{shahmansoori2017position}. With  UE mobility, the derivatives involving radial velocity $v_l$  of the $l$th path can be expressed as
\begin{align}
    \frac{\partial v_l}{\partial \pv_0} & =
    \begin{cases}
    (\frac{\partial{\tv_{\text{U},0}}}{\partial \pv_0})^\top\vv 
    = \frac{v_0\tv_{\text{U}, 0}-\vv }{d_\text{0}} & l=0\\
    (\frac{\partial{\tv_{\text{U}, l}}}{\partial \pv_0})^\top\vv 
    = \frac{v_l\tv_{\text{U},l}-\vv }{d_{l,2}} & l>0
    \label{eq:derivative_vl_PU}
    \end{cases},
    \\
    \frac{\partial{v_l}}{\partial{\pv_{l'}}} & =
    \begin{cases}
    (\frac{\partial{\tv_{\text{U},l}}}{\partial \pv_l})^\top\vv 
    = -\frac{v_l\tv_{\text{U},l}-\vv }{d_{l,2}} & l=l'>0\\
    0 & \text{others}
    \end{cases},
    \label{eq:derivative_vl_Pl}
    \\
    \frac{\partial{v_l}}{\partial{B}} & = 0,
    \quad
    \frac{\partial v_l}{\partial \vv}  =
    \tv_{\text{U},l}.
    \label{eq:derivative_vl_VU}
\end{align}
Finally, 
 we can define the \ac{peb}, \ac{meb}, \ac{meb} and \ac{veb} as
\begin{align}
\text{PEB} & = \sqrt{\trace([\mathrm{CRB}]_{1:2, 1:2})},
\label{eq:PEB}
\\
\text{MEB}_l & = \sqrt{\trace([\mathrm{CRB}]_{(2l+1):(2l+2), (2l+1):(2l+2)})},
\label{eq:MEB}
\\
\text{CEB} & = \sqrt{[\mathrm{CRB}]_{(2L+3),(2L+3)}},
\label{eq:CEB}
\\
\text{VEB} & = \sqrt{\trace([\mathrm{CRB}]_{(2L+4):(2L+5), (2L+4):(2L+5)})},
\label{eq:VEB}
\end{align}
where $\trace(\cdot)$ returns the trace of a matrix, and $[\cdot]_{i,j}$ is getting the element in the $i$th row, $j$th column of a matrix. The bounds from~\eqref{eq:PEB}--\eqref{eq:VEB} will be used to evaluate the localization and mapping performance.

\subsection{FIM Analysis for Varying Speed}
To evaluate the effect of speed on localization and mapping performance, we re-order the rows and columns of the matrix $\Jm_\text{S}$ (take $L=2$ for illustration) as
\begin{align}
    & \Jm_\mathrm{R} =     
    \left[
    \begin{array}{c : c}
        \Am & \Bm \\ \hdashline
        \Om & \Dm 
    \end{array}
    \right] = \label{eq:jacobian_reorder}
    \\
    & \left[\begin{array}{c c c c c c : c c c}
        \vspace{1mm}
        \frac{\partial \theta_0}{\partial \pv_0} 
        & \frac{\partial \tau_0}{\partial \pv_0} 
        & \mathbf{0}_{2}
        & \frac{\partial \tau_1}{\partial \pv_0} 
        & \mathbf{0}_{2}
        &\frac{\partial \tau_2}{\partial \pv_0}
        & \frac{\partial v_0}{\partial \pv_0} 
        & \frac{\partial v_1}{\partial \pv_0} 
        & \frac{\partial v_2}{\partial \pv_0} 
        \\        
        \vspace{1mm}
        \mathbf{0}_{2} & \mathbf{0}_{2}
        & \frac{\partial \theta_1}{\partial \pv_\text{1}} 
        & \frac{\partial \tau_1}{\partial \pv_\text{1}} 
        & \mathbf{0}_{2}
        & \mathbf{0}_{2}
        & \mathbf{0}_{2}
        & \frac{\partial v_1}{\partial \pv_\text{1}} 
        & \mathbf{0}_{2}
        \\
        \mathbf{0}_{2} & \mathbf{0}_{2}
        & \mathbf{0}_{2}
        & \mathbf{0}_{2}
        & \frac{\partial \theta_1}{\partial \pv_\text{2}} 
        & \frac{\partial \tau_1}{\partial \pv_\text{2}} 
        & \mathbf{0}_{2}
        & \mathbf{0}_{2}
        & \frac{\partial v_2}{\partial \pv_\text{2}} 
        \\
        \vspace{1mm}
        0 & \frac{\partial \tau_0}{\partial B} 
        & 0
        & \frac{\partial \tau_1}{\partial B} 
        & 0
        & \frac{\partial \tau_2}{\partial B}
        & 0
        & 0
        & 0
        \\
        \hdashline
        \vspace{1mm}
        \mathbf{0}_{2}
        & \mathbf{0}_{2}
        & \mathbf{0}_{2}
        & \mathbf{0}_{2}
        & \mathbf{0}_{2}
        & \mathbf{0}_{2}
        & \frac{\partial v_0}{\partial \vv} 
        & \frac{\partial v_1}{\partial \vv} 
        & \frac{\partial v_2}{\partial \vv}
    \end{array}
    \right],\notag 
\end{align}
where $\mathbf{0}_{2}$ is a $2\times 1$ zero vector. By defining a direction vector $\tv_v = [\cos(\theta_v), \sin(\theta_v)]^\top$ and speed $v=\Vert \vv \Vert$ such that $\vv = v\tv_v$, we noticed that the submatrices $\Am$, $\Om$ and $\Dm$ will not change with speed $v$ if the velocity direction $\tv_v$ is fixed. Based on this observation, we define a matrix $\bar \Bm = \Bm/v$. In the following, we show that the PEB, MEB and CEB will reduce and saturate with an increased $v$, showing the velocity can improve these bounds to a certain level.

With reasonable assumptions, such as the AODs have no spatial correlation, delays of different paths can be resolved, the FIM of the delays and AODs in stationary scenario can be approximated as a diagonal matrix~\cite{mendrzik2018harnessing}. We make further approximation by assuming the radial velocity is independent of other channel parameters and the reordered FIM $\Fm$ can be formulated from $\mathbf{I}(\gammav)$ as 
\begin{equation}
    \Fm(\gammav) = \text{diag}(\mathbf{I}(\gammav)^{-1})^{-1},
    \label{eq:approximated_fim}
\end{equation}
where $\text{diag}(\cdot)$ is the operation that keeps only the diagonal elements of a matrix (i.e., to form a diagonal matrix). We further segment $\Fm(\gammav)$ into two diagonal matrix as $\Fm(\gammav) = \text{blkdiag}(\Fm_1, \Fm_2)$ and the matrices $\Fm_1\in \mathbb{R}^{(2L+2)\times(2L+2)}$ and $\Fm_2\in \mathbb{R}^{(L+1)\times(L+1)}$ contain the variance of all the channel parameters that can be described as
$\Fm_1 = \text{diag}(1/\sigma^2_{\theta_0}, 1/\sigma^2_{\tau_0}, 1/\sigma^2_{\theta_1}, 1/\sigma^2_{\tau_1}, 1/\sigma^2_{\theta_2}, 1/\sigma^2_{\tau_2})$, and $\Fm_2 = \text{diag}(1/\sigma^2_{v_0}, 1/\sigma^2_{v_1}, 1/\sigma^2_{v_2})$.
The approximated FIM can then be formed as $\Jm_\text{S}\Fm(\gammav) \Jm_\text{S}^\top$, and we will show that the error of this approximation is in the simulation Section~\ref{sec:doppler_enhanced_localization}.

The EFIM of the state vector containing positions (both UE and IPs) and clock offset can be expressed as
\begin{equation}
\begin{split}
    \Em_\text{s} = & \underbrace{\Am\Fm_1\Am^\top}_{\Am_\text{S} \text{(stationary info)}} 
    + \underbrace{v^2\bar \Bm\Fm_2\bar \Bm^\top}_{\Bm_\text{G}=v^2\bar{\Bm}_\text{G} \ \text{(mobility gain)}} \\
    & - \underbrace{v^2\bar\Bm\Fm_2\Dm^\top (\Dm\Fm_2\Dm^\top )^{-1} \Dm\Fm_2\bar\Bm^\top}_{\Bm_\text{L}=v^2\bar{\Bm}_\text{L}\ \text{(mobility estimation loss)}}.
    \label{eq:EFIM_1}
\end{split}
\end{equation}
The matrix $\Am_\text{S}$ is the FIM for stationary UE, the matrix $\Bm_\text{G}$ is the information gain with UE mobility if the velocity vector $\vv$ is known, and $\Bm_\text{L}$ is the information loss with unknown UE velocity. 

Similarly, the EFIM of the UE velocity can be expressed as
\begin{equation}
    \Em_\text{v} = \underbrace{\Dm\Fm_2\Dm^\top}_{=\Dm_0 \text{(velocity info)}} 
    - \underbrace{v^2\Dm\Fm_2\Bm^\top (\Am_0+v^2\bar\Bm_\text{G})^{-1} \Bm\Fm_2\Dm^\top}_{=\Dm_\text{L}=v^2\bar{\Dm}_\text{L}\text{(unknown estimation loss)}}.
    \label{eq:EFIM_2}
\end{equation}
where $\Dm_0$ contains the velocity from the radial velocity estimation and $\Dm_\text{L}$ is the information loss due to other unknown parameters.

\begin{lemma}
\label{lemma:EFIM_position}
The EFIM of the UE and IP positions $\Em_\text{p}$ and the EFIM of  clock offset $E_\text{c}$ are given by
\begin{align}
    \Em_\text{p} & = \Am_0 + v^2\Bm_0 = \Am_1 - \Am_2 A_4^{-1}\Am_3 + v^2\Bm_0,
    \label{eq:EFIM_position}
    \\
    E_\text{c} & = A_4 - \av^\top_2(\Am_1 + v^2\Bm_0)^{-1}\av_2,
\end{align}
where $\Am_1 = [\Am_\text{S}]_{1:(2L+2), 1:(2L+2)}$, $\av_2 = [\Am_\text{S}]_{1:(2L+2), (2L+3)}$, $A_4 = [\Am_\text{S}]_{(2L+3), (2L+3)}$ are the submatrices of the matrix $\Am_\text{S}$, and $\Bm_0 = [\bar\Bm_\text{G} - \bar\Bm_\text{L}]_{1:(2L+2), 1:(2L+2)}$.
\end{lemma}

\begin{proof}
By segmenting $\Bm$ into a $\tilde \Bm = [\Bm]_{1:(2L+2), 1:(L+1)}$ and a zero vector $\mathbf{0}_{1\times (L+1)}$, we can see that the last row and last column of matrices $\Bm_\text{G}$, $\Bm_\text{L}$ are all zeros, and Lemma~\ref{lemma:EFIM_position} can be obtained based on the matrix inverse lemma.
\end{proof}

\begin{prop}
When the speed $v\to 0$, the localization problem is not solvable.
\end{prop}
\begin{proof}
We notice that the matrix $\tilde \Am = [\Am]_{1:(2L+2), 1:(2L+2)}$ (first $2L+2$ rows, without the row containing clock offset) from~\eqref{eq:jacobian_reorder} is a block uppler diagonal matrix. Since the matrix $[{\partial \theta_l}/{\partial \pv_l}, {\partial \tau_l}/{\partial \pv_l}]$ is of rank 2 (nonzero determinant), $\tilde \Am$ has full rank. Then, $\Am$ is of rank $(2L+2)$ and hence, $\Am_\text{S} =\Am\Fm_1 \Am^\top $ is of rank $(2L+2)$. As a consequence, the matrix $\Am_\text{S}$ is of rank $2L+2$ (which is not full rank due to the unknown clock offset), indicating the localization and mapping can not be performed without UE mobility.\footnote{{Note that the VEB does not grow with increasing UE/IP position estimation error. Instead, it is more affected by the speed. For example, when the speed (norm of the velocity) is small, the radial velocity estimations are also small, resulting a small velocity estimation error regardless of how large the UE/IP position estimation errors are (as shown in the simulation). However, this velocity estimation (although with a small error bound) does not help in solving the localization problem since the scenario is almost stationary.}}
\end{proof}
This proposition is clearly congruent with the observations from Section \ref{sec:doppler_assisted_localization}.
Moreover,  the scenario with $v \to 0$ is equivalent to the setup in \cite{fascista2019millimeter,fascista2022ris}, where the localization problem was solved by assuming perfect knowledge of the clock bias $B$ \cite{fascista2019millimeter} or a reconfigurable intelligent surface \cite{fascista2022ris}. 
\begin{prop}
\label{prop:EP_Ec_constant}
When the speed $v\to \infty$, the PEB, MEB, and CEB will converge to a constant value, whereas the VEB keeps increasing with $v$.
\end{prop}
\begin{proof}
Based on the lemma derived in~\cite{miller1981inverse}, stating if $\Qm$ has rank one, $\Em$ and $\Em+\Qm$ are nonsingular, then 
\begin{equation}
    (\Em+v\Qm)^{-1} = \Em^{-1}-\frac{v}{1+v\text{tr}(\Qm\Em^{-1})}\Em^{-1}\Qm\Em^{-1}.
    \label{eq:miller_lemma}
\end{equation}
Therefore, we can decompose $\Bm_0$ into a summation of several rank-1 matrices $\Bm_1, \Bm_2, \ldots, \Bm_{R_B}$ by using SVD, where $R_\text{B}$ is the rank of $\Bm_0$.
We can see the improvement of localization and mapping performance will saturate as $\Em^{-1}_\text{p}(v\to{\infty}) = \Em^{-1}_{\text{p},R_\text{B}}$, which can be obtained recursively from~\eqref{eq:miller_lemma} as
\begin{equation}
    \Em_{\text{p},i}^{-1} = 
    \begin{cases}
    \Am_0^{-1}-\frac{1}{\text{tr}(\Bm_1\Am_0^{-1})}\Am_0^{-1}\Bm_1 \Am_0^{-1} & i = 1,\\
    \Em_{\text{p},i-1}^{-1}-\frac{1}{\text{tr}(\Bm_i\Em_{\text{p},i-1}^{-1})}\Em_{\text{p},i-1}^{-1}\Bm_i\Em_{\text{p},i-1}^{-1} & i \ne 1,
    \end{cases}
    \label{eq:miller_lemma_iterative}
\end{equation}
where $\Em_{\text{p},i} = \Am_0 + \sum_{j=1}^{i}v^2\Bm_j$. Since we have shown in Lemma~\eqref{prop:EP_Ec_constant} that $(\Am_0 + v^2\Bm_0)^{-1}$ is getting close to a constant matrix when $v\to\infty$, $E_\text{c}$ is also close to a constant.
Considering $\Am_0+\Bm_\text{G}=\Am_0+v^2\bar\Bm_\text{G}$ is a constant matrix with large $v$, we can see the VEB keeps increasing with speed $v$ as the estimation loss $\Dm_\text{L}$ in~\eqref{eq:EFIM_2} increases linearly with $v^2$. Thus, Proposition~\ref{prop:EP_Ec_constant} is proved, which is also validated in the simulation results in Section \ref{sec:doppler_enhanced_localization}.
\end{proof}

\subsection{Localization and Mapping Algorithm}
With a high dimension of unknowns (size of state vector $\sv$), it is not practical to perform \ac{mle} and estimate all the unknowns at the same time (e.g., using gradient descent). {Here, similar to the ad-hoc estimator in~\cite{nazari2022mmwave}, we propose an algorithm that simplifies localization and mapping tasks into a 1D search problem.}

Consider the LOS channel has the strongest signal strength, we search along the direction of the estimated \ac{aod} $\hat \theta_{0}$ at BS, which is equivalent to the direction vector $\hat \tv_{\text{B},0}$ that can be calculated from~\eqref{eq:angle_at_BS}. For a given UE candidate position\footnote{We use the notation $\tilde \cdot$ to indicate the variables depending on the candidate $\tilde d_0$, and use the notation $\hat \cdot$ to represent the estimated channel parameters.} $\tilde \pv_0 = \tilde d_0 \hat \tv_0$ on the line $\pv = d \hat \tv_0$, the clock offset (for this specific UE candidate $\tilde \pv_0$) can be estimated as $\tilde B = c\hat \tau_0 - \Vert \tilde\pv_0 \Vert$ and the propagation distance of the NLOS channel can be obtained as $\tilde d_l = c\hat \tau_l - \tilde B$. We further define an intermediate variable $e_l$ for the $l$th path as
\begin{align}
    \tilde e_{l} & = (\tilde \pv_0-\pv_\text{B})^\top \tilde\tv_{\text{B},l} = \tilde \pv_0^\top \tilde\tv_{\text{B},l},
    \label{eq:intermediate_el}
\end{align}
where $\tilde \tv_{\text{B}, l} = \hat \tv_{\text{B}, l}$ is the estimated direction vector from angle $\hat \theta_l$ based on~\eqref{eq:angle_at_BS}. The position of the $l$th \ac{ip} can be obtained from $\tilde d_0^2 - \tilde e_{l}^2 + (\tilde d_{l,1} - \tilde e_{l})^2 = (\tilde d_l - \tilde d_{l,1})^2$ as
\begin{equation}
    \tilde \pv_{l} = \pv_\text{B} + \tilde d_{l,1} \tilde \tv_{\text{B}, l} = \tilde d_{l,1} \tilde \tv_{\text{B}, l},  \ \ \tilde d_{l,1} = \frac{\tilde d_{0}^2 -\tilde d_l^2 }{2(\tilde e_{l} - \tilde d_l)}.
\end{equation}

After obtaining all the position of the IPs, we can then obtain the direction vector $\tilde \tv_{\text{U},l}$ for each path based on \eqref{eq:state_velocity} and hence the velocity can be estimated using a least squares method as
\begin{equation}
    \tilde \vv = (\tilde \Xm^\top \tilde \Xm)^{-1}\tilde \Xm^\top \hat \bv,
    \label{eq:hat_velocity}
\end{equation}
where $\tilde \Xm = [\tilde \tv_{\text{U}, 0}, \ldots,\tilde \tv_{\text{U}, L}]^\top$, $\hat \bv = [\hat v_0, \ldots, \hat v_L]^\top$. The estimated velocity $\tilde \vv$ is the velocity that fits current UE candidate $\tilde \pv_0$ and radial velocity estimation $\hat \bv$ the best. Since  $\tilde \Xm$ is a function of $d_0$, the residual error can be expressed as $\varepsilon(d_0)=\Vert \tilde \Xm \tilde \vv - \hat \bv\Vert$, from which 
\begin{equation}
    d_0^* = \arg \min_{d_0}\varepsilon(d_0).
    \label{eq:cost_function}
\end{equation}
Finally, the estimated position of the UE can be obtained as $\pv_0^* = d_0^*\hat \tv_0$, 
and the remaining unknowns $\pv_l^*$, $B^*$ can be obtained similarly from~\eqref{eq:intermediate_el} to~\eqref{eq:hat_velocity}.
Since the positions of the IPs are obtained from the estimated channel parameters $d_l$ and $\theta_{l}$, equation~\eqref{eq:cost_function} is identical to $\argmin_{d_0} ||\tilde \gammav(\tilde \pv_{0}) - \hat \gammav||$. Further improvement can adopt weighted least square with the covariance matrix of the estimated radial velocity at each path, or refine the results using gradient descent after getting the initial result from the proposed 1D search.

\section{Numerical Results}

\subsection{Simulation Parameters}
We consider a 2D downlink scenario with a single-antenna UE and a BS with a 16-element ULA lies on the x-axis. The pilot signal $x_{g,k}$ is chosen with a constant amplitude and random phase. The channel gain for each path is calculated as $\rho_0 =  \frac{\lambda}{4\pi d_0}e^{-j \frac{2\pi}{\lambda} d_0}$ for the LOS path and $\rho_l =  \sqrt{\frac{c_{l}}{4\pi}}\frac{\lambda}{4\pi d_{l,1}d_{l,2}} e^{-j \frac{2 \pi}{\lambda} d_l}$ for the $l$th NLOS path, where $c_{l}$ is the radar cross-section (RCS).
The default simulation parameters can be found in  Table~\ref{table:Simulation_parameters}. 


\vspace{-5mm}
\begin{table}[ht]
\scriptsize
\centering
\caption{Default Simulation Parameters}
\renewcommand{\arraystretch}{1.25}
\begin{tabular} {c | c }
    \hthickline
    \textbf{Types} & \textbf{Simulation Parameters}\\
    \hline
    BS Position & {$\pv_\text{B} = [0, 0]^\top$}   \\
    \hline
    UE Position & {$\pv_0 = [5, 2]^\top$}   \\
    \hline
    BS Array Size & {$N_{\text{B}} = 16$}   \\
    \hline
    IP Positions & $\pv_1 = [-6, 8]^\top$, $\pv_2=[8, 6]^\top$ \\
    \hline
    Measurement Interval & $T_\text{int} = \unit[1]{ms}$ \\
    \hline
    RCS coefficients & $c_l = \unit[10]{m^2}$ \\
    \hline
    Carrier Frequency & {$f_c = \unit[28]{GHz}$}  \\
    \hline
    Bandwidth & {$W = \unit[400]{MHz}$}  \\
    \hline
    Number of Transmissions & {$G = 20$}  \\
    \hline
    Number of Subcarriers & {$K = 20$} \\
    \hline
    Average Transmission Power & {$P = 30$ dBm}  \\
    \hline
    Noise PSD & {$N_0 = \unit[-173.855]{dBm/Hz}$} \\
    \hline
    Noise Figure & {$\unit[10]{dBm}$} \\
    \hline
\end{tabular}
\renewcommand{\arraystretch}{1}
\label{table:Simulation_parameters}\vspace{-5mm}
\end{table}

\subsection{Performance Bounds Results}
\label{sec:doppler_enhanced_localization}
Based on the analysis from Section~\ref{sec:doppler_assisted_localization}, the minimum number of IPs that can support localization and mapping under UE mobility is 2. We use the parameters provided in Table II and visualize PEB, MEB, CEB, and VEB with different UE positions; the results are shown in Fig.~\ref{fig:heatmap}. We can see that the PEB, CEB, and OEB are showing a similar pattern, and a low error bound can be found in the convex hull formed by BS and IPs. In contrast, low VEBs are seen in the UE positions where the IP is aligned with the velocity direction.
\begin{figure}[t]
\begin{minipage}[b]{0.48\linewidth}
    \centering
    \centerline{\includegraphics[width=0.98\linewidth]{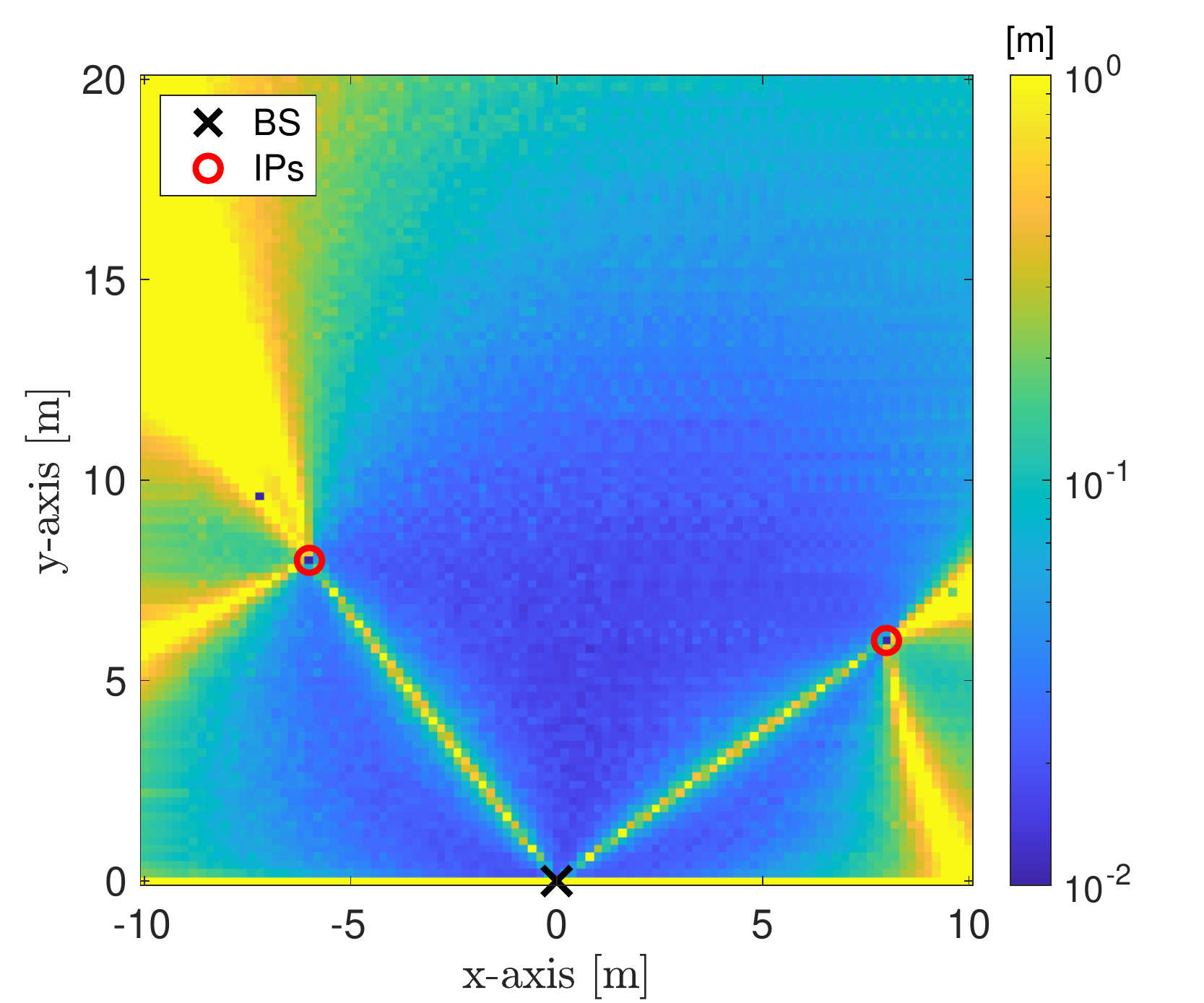}}
    \centerline{(a) PEB}
\end{minipage}
\begin{minipage}[b]{0.48\linewidth}
    \centering
    \centerline{\includegraphics[width=0.98\linewidth]{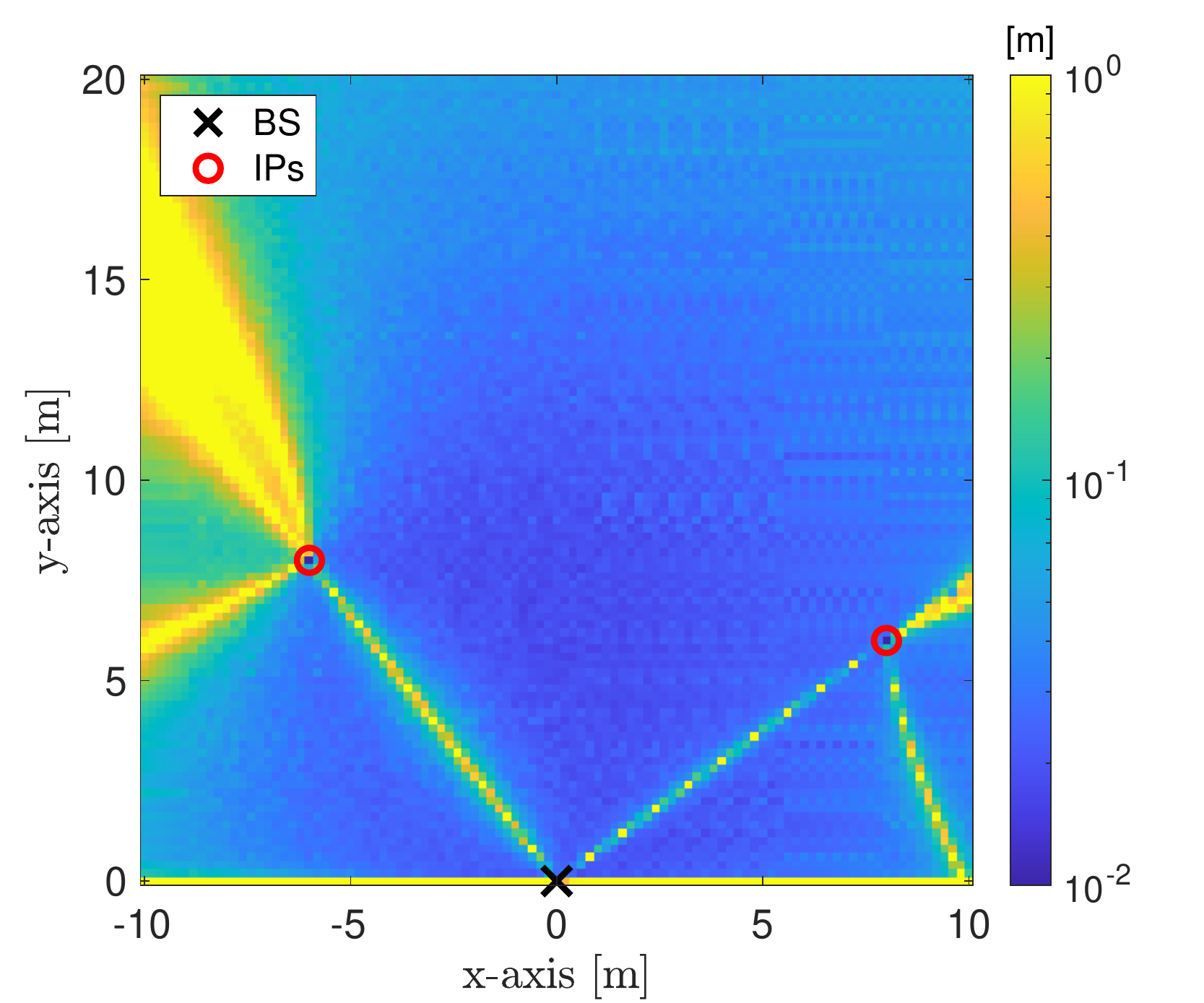}}
    \centerline{(b) MEB}
\end{minipage}
\begin{minipage}[b]{0.48\linewidth}
    \centering
    \centerline{\includegraphics[width=0.98\linewidth]{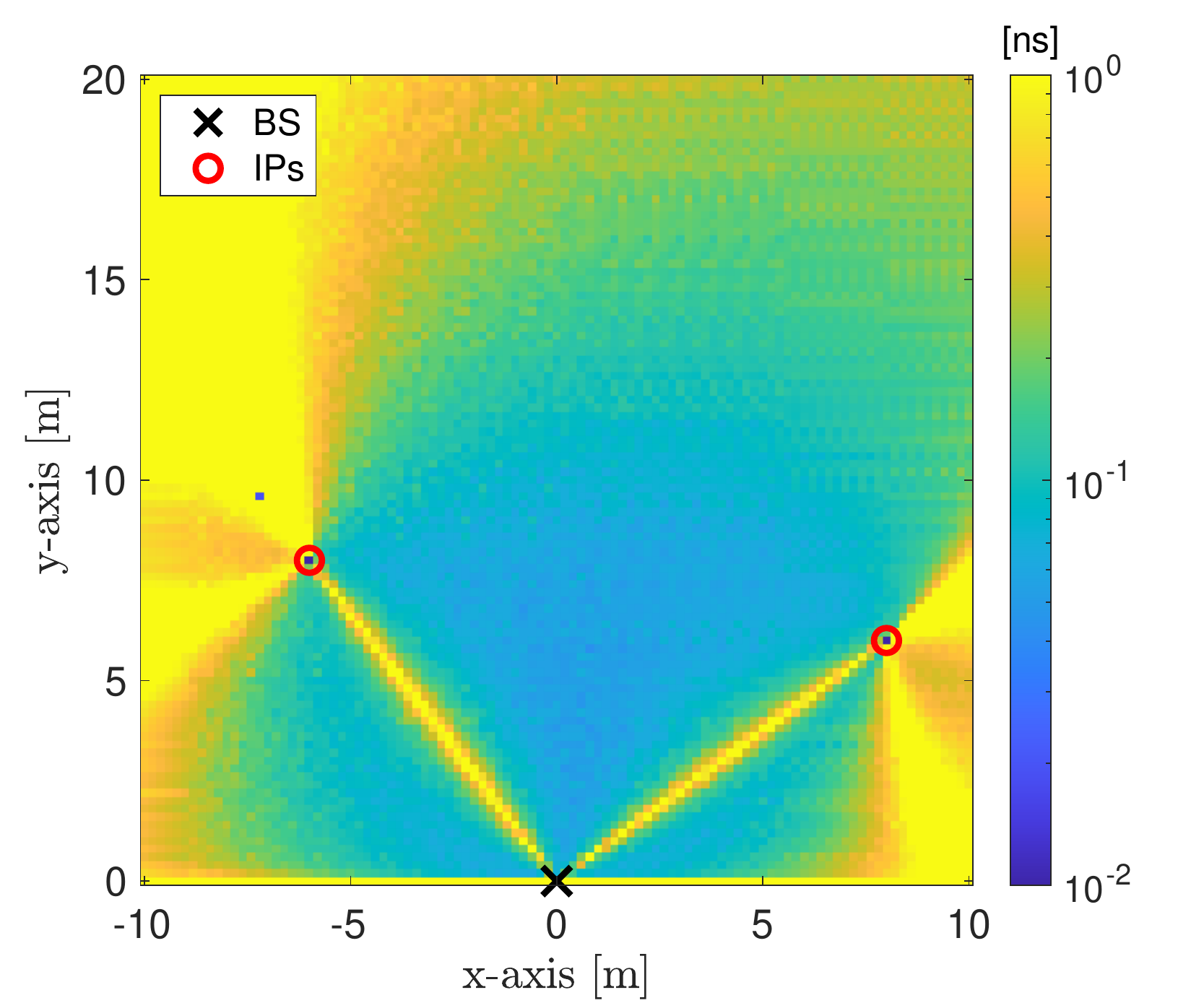}}
    \centerline{(c) CEB}
\end{minipage}
\hspace{1mm}
\begin{minipage}[b]{0.48\linewidth}
    \centering
    \centerline{\includegraphics[width=0.98\linewidth]{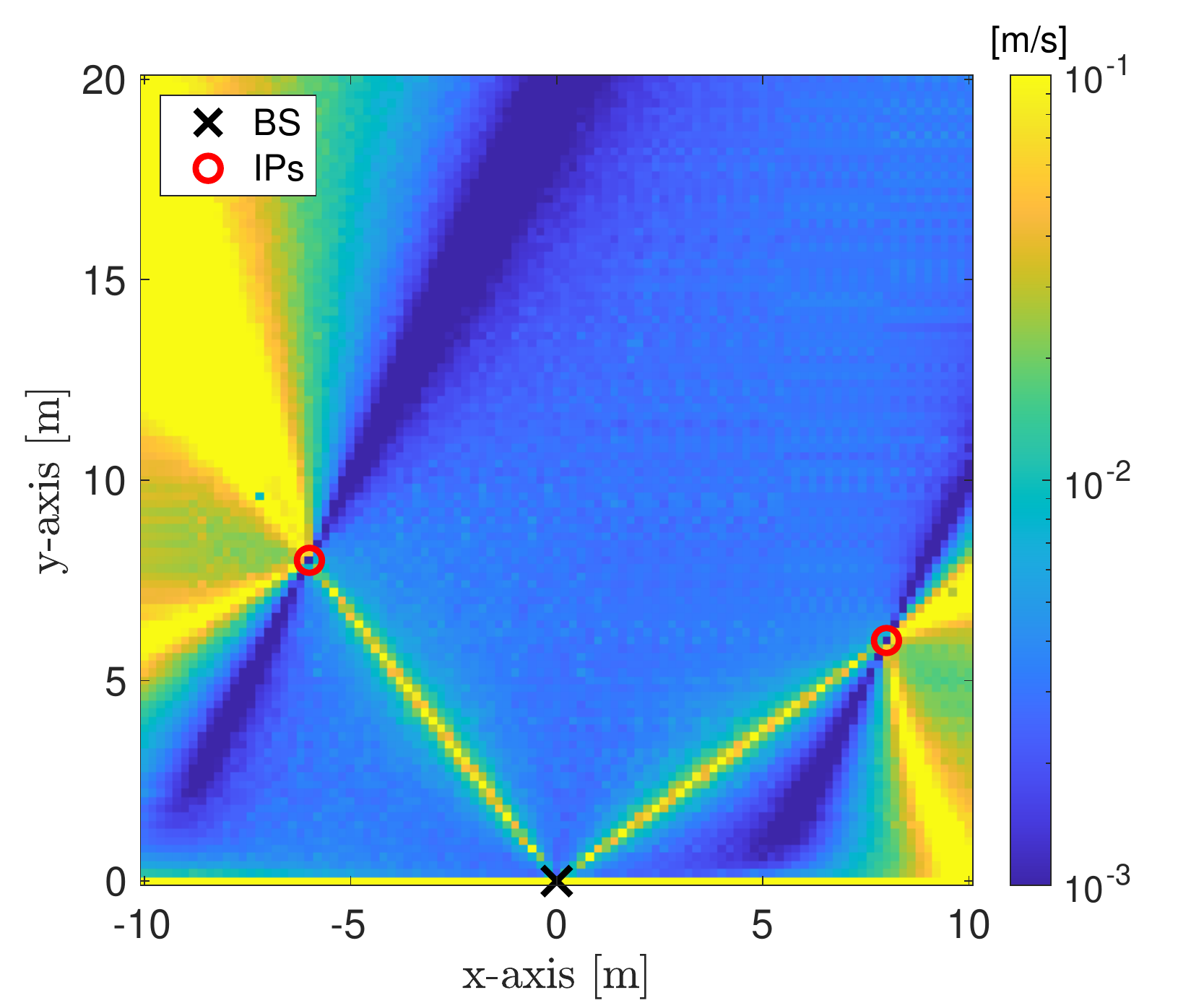}}
    \centerline{(d) VEB}
\end{minipage}
\caption{Visualization of PEB, MEB, CEB, and VEB for different UE positions. We can see that worse performance is shown in the area that lie between the BS and scattering points due to the unresolvable paths. Regarding the VEB, when the IPs is on the direction of the velocity ($[1, 2]^\top$), more accurate estimation can be achieved.}\vspace{-3mm}
\label{fig:heatmap}
\end{figure}

We evaluate how the PEB and VEB are effected by changing UE speed from $\unit[0.001]{m/s}$ to $\unit[10]{m/s}$ with a fixed UE position $\pv_0 = [5,2]^\top$. Two scenarios are evaluated, namely, velocity direction $[1, 2]^\top/\sqrt{5}$ (scenario 1) and $[2, 1]^\top/\sqrt{5}$ (scenario 2). The approximated error bound from~\eqref{eq:approximated_fim} and the PEB when $v\to\infty$ are shown in the figure. With the increase of speed $v$, the PEBs are getting lower and saturate at around $\unit[3]{m/s}$. As for the estimation of velocity, VEB keeps increasing with speed. Note that in reality the speed cannot be too large as the channel may not be coherent.

\begin{figure}[htb]
\begin{minipage}[b]{0.98\linewidth}
  \centering
%
%
\begin{tikzpicture}
[scale=1\columnwidth/10cm,font=\normalsize]

\begin{axis}[%
width=8cm,
height=3.6cm,
at={(0,0)},
scale only axis,
xmode=log,
xmin=0.001,
xmax=10,
xlabel style={font=\footnotesize\color{white!15!black}, yshift=1 ex},
xlabel={UE Speed [m/s]},
ymode=log,
ymin=0.01,
ymax=100,
yminorticks=true,
ylabel style={font=\footnotesize\color{white!15!black}, yshift= -1.5 ex},
ylabel={PEB and MEB [m]},
axis background/.style={fill=white},
xmajorgrids,
ymajorgrids,
legend style={font=\scriptsize, at={(.65, 0.35)}, anchor=south west, legend cell align=left, align=left, draw=white!5!black, legend columns=1}
]
\addplot [color=blue, dashed, line width=1.0pt, mark=o, mark options={solid, blue}]
  table[row sep=crcr]{%
0.001	2.05231675530402\\
0.00193069772888325	1.05702578152524\\
0.00372759372031494	0.612890342754208\\
0.00719685673001152	0.362506148214594\\
0.0138949549437314	0.264072603245197\\
0.0268269579527973	0.220569221209575\\
0.0517947467923121	0.199536209353877\\
0.1	0.199824688022964\\
0.193069772888325	0.194746953528184\\
0.372759372031494	0.18873168029308\\
0.719685673001152	0.186279292962096\\
1.38949549437314	0.184642094806521\\
2.68269579527972	0.184714131500351\\
5.17947467923121	0.183269307544904\\
10	0.182234050150006\\
};
\addlegendentry{PEB-1}

\addplot [color=blue, line width=1.0pt]
  table[row sep=crcr]{%
0.001	2.0092476963073\\
0.00193069772888325	1.0345923113594\\
0.00372759372031494	0.589684389231355\\
0.00719685673001152	0.345578766480947\\
0.0138949549437314	0.242088590455215\\
0.0268269579527973	0.205750492011588\\
0.0517947467923121	0.192249141931005\\
0.1	0.186000644740248\\
0.193069772888325	0.177912916728071\\
0.372759372031494	0.184537921667066\\
0.719685673001152	0.179927551264861\\
1.38949549437314	0.179892858087137\\
2.68269579527972	0.178972653166552\\
5.17947467923121	0.178943219001047\\
10	0.178225671470211\\
};
\addlegendentry{PEB-1 (Approx.)}

\addplot [color=black, dashed, line width=1.0pt]
  table[row sep=crcr]{%
0.001	0.192199332569125\\
0.00193069772888325	0.192199332569125\\
0.00372759372031494	0.192199332569125\\
0.00719685673001152	0.192199332569125\\
0.0138949549437314	0.192199332569125\\
0.0268269579527973	0.192199332569125\\
0.0517947467923121	0.192199332569125\\
0.1	0.192199332569125\\
0.193069772888325	0.192199332569125\\
0.372759372031494	0.192199332569125\\
0.719685673001152	0.192199332569125\\
1.38949549437314	0.192199332569125\\
2.68269579527972	0.192199332569125\\
5.17947467923121	0.192199332569125\\
10	0.192199332569125\\
};
\addlegendentry{PEB-1 ($v=\infty$)}

\addplot [color=red, dashed, line width=1.0pt, mark=square, mark options={solid, red}]
  table[row sep=crcr]{%
0.001	0.335160402598675\\
0.00193069772888325	0.176013784916722\\
0.00372759372031494	0.0956843294154557\\
0.00719685673001152	0.0552367340110614\\
0.0138949549437314	0.0368516556043407\\
0.0268269579527973	0.029394797303883\\
0.0517947467923121	0.0284084771849548\\
0.1	0.0279990754833024\\
0.193069772888325	0.0279194482216571\\
0.372759372031494	0.0273295392635597\\
0.719685673001152	0.0272126643995844\\
1.38949549437314	0.026689798288475\\
2.68269579527972	0.0268578789337406\\
5.17947467923121	0.0270279365893414\\
10	0.0268748996671874\\
};
\addlegendentry{PEB-2}

\addplot [color=red, line width=1.0pt]
  table[row sep=crcr]{%
0.001	0.328779252851612\\
0.00193069772888325	0.173141947438039\\
0.00372759372031494	0.0940337525732492\\
0.00719685673001152	0.0540469277971233\\
0.0138949549437314	0.0362997517134831\\
0.0268269579527973	0.0292121095387503\\
0.0517947467923121	0.0281318500227459\\
0.1	0.027100412373854\\
0.193069772888325	0.026734551848344\\
0.372759372031494	0.0267978974560756\\
0.719685673001152	0.0267429071591288\\
1.38949549437314	0.0263267884322496\\
2.68269579527972	0.0263952954190414\\
5.17947467923121	0.0267191971384466\\
10	0.0264154778102693\\
};
\addlegendentry{PEB-2 (Approx.)}

\addplot [color=black, dashdotted, line width=1.0pt]
  table[row sep=crcr]{%
0.001	0.0266498799116434\\
0.00193069772888325	0.0266498799116434\\
0.00372759372031494	0.0266498799116434\\
0.00719685673001152	0.0266498799116434\\
0.0138949549437314	0.0266498799116434\\
0.0268269579527973	0.0266498799116434\\
0.0517947467923121	0.0266498799116434\\
0.1	0.0266498799116434\\
0.193069772888325	0.0266498799116434\\
0.372759372031494	0.0266498799116434\\
0.719685673001152	0.0266498799116434\\
1.38949549437314	0.0266498799116434\\
2.68269579527972	0.0266498799116434\\
5.17947467923121	0.0266498799116434\\
10	0.0266498799116434\\
};
\addlegendentry{PEB-2 ($v=\infty$)}

\end{axis}

\end{tikzpicture}%
    \vspace{-0.8cm}
    \centerline{(a) PEB} \medskip
\end{minipage}
\hfill
\begin{minipage}[b]{0.98\linewidth}
  \centering
%
%
\begin{tikzpicture}
[scale=1\columnwidth/10cm,font=\normalsize]

\begin{axis}[%
width=8cm,
height=3.6cm,
at={(0, 0)},
scale only axis,
xmode=log,
xmin=0.001,
xmax=10,
xlabel style={font=\footnotesize\color{white!15!black}, yshift=1 ex},xlabel={UE Speed [m/s]},
ymode=log,
ymin=0.00001,
ymax=0.1,
yminorticks=true,
ylabel style={font=\footnotesize\color{white!15!black}, yshift= -1.5 ex},
ylabel={VEB [m/s]},
axis background/.style={fill=white},
xmajorgrids,
ymajorgrids,
legend style={font=\scriptsize, at={(.05, 0.7)}, anchor=south west, legend cell align=left, align=left, draw=white!5!black, legend columns=2}
]
\addplot [color=blue, dashed, line width=1.0pt, mark=o, mark options={solid, blue}]
  table[row sep=crcr]{%
0.001	9.99683745537611e-05\\
0.00193069772888325	9.97095611163948e-05\\
0.00372759372031494	0.000112112924339609\\
0.00719685673001152	0.000133114773752927\\
0.0138949549437314	0.000193039707963511\\
0.0268269579527973	0.000315677001246084\\
0.0517947467923121	0.00055358800071274\\
0.1	0.00107855356708907\\
0.193069772888325	0.0020159165665039\\
0.372759372031494	0.00371385179168138\\
0.719685673001152	0.00709258900114552\\
1.38949549437314	0.0135427473703326\\
2.68269579527972	0.0262745113497429\\
5.17947467923121	0.0500224383647239\\
10	0.0959705252547061\\
};
\addlegendentry{VEB-1}

\addplot [color=blue, line width=1.0pt]
  table[row sep=crcr]{%
0.001	9.74213733042492e-05\\
0.00193069772888325	9.68937422593945e-05\\
0.00372759372031494	0.0001070429031395\\
0.00719685673001152	0.000124820124679243\\
0.0138949549437314	0.000173838440472195\\
0.0268269579527973	0.000290308925370555\\
0.0517947467923121	0.000529169061291027\\
0.1	0.000991768928208965\\
0.193069772888325	0.00181257154568342\\
0.372759372031494	0.00362167950877296\\
0.719685673001152	0.00681518332603513\\
1.38949549437314	0.013146881316181\\
2.68269579527972	0.0253509004091447\\
5.17947467923121	0.0486834140817337\\
10	0.0935134819339148\\
};
\addlegendentry{VEB-1 (Approx)}

\addplot [color=red, dashed, line width=1.0pt, mark=square, mark options={solid, red}]
  table[row sep=crcr]{%
0.001	1.36720386095276e-05\\
0.00193069772888325	1.36931881540449e-05\\
0.00372759372031494	1.40436640053505e-05\\
0.00719685673001152	1.53387074074337e-05\\
0.0138949549437314	1.87502306165807e-05\\
0.0268269579527973	2.68102358146405e-05\\
0.0517947467923121	4.69440005738694e-05\\
0.1	9.38182829363082e-05\\
0.193069772888325	0.000184618931005487\\
0.372759372031494	0.000314996697241658\\
0.719685673001152	0.000613668275969247\\
1.38949549437314	0.00114667905754419\\
2.68269579527972	0.0022280747812471\\
5.17947467923121	0.00423783029274279\\
10	0.00828007603289085\\
};
\addlegendentry{VEB-2}

\addplot [color=red, line width=1.0pt]
  table[row sep=crcr]{%
0.001	1.32099711255812e-05\\
0.00193069772888325	1.32652872326956e-05\\
0.00372759372031494	1.35615407081692e-05\\
0.00719685673001152	1.46220398071805e-05\\
0.0138949549437314	1.81022147605791e-05\\
0.0268269579527973	2.67447513479014e-05\\
0.0517947467923121	4.6327897467805e-05\\
0.1	8.49061050375584e-05\\
0.193069772888325	0.000159489094560843\\
0.372759372031494	0.000302274614331003\\
0.719685673001152	0.000589946373996789\\
1.38949549437314	0.00110630020910665\\
2.68269579527972	0.00215970436274823\\
5.17947467923121	0.00414415058342096\\
10	0.00799295725907998\\
};
\addlegendentry{VEB-2 (Approx)}
\end{axis}

\end{tikzpicture}%
    \vspace{-0.8cm}
    \centerline{(b) VEB}
\end{minipage}
\caption{PEB and VEB change with different UE speed. The velocity directions are $[1, 2]^\top/\sqrt{5}$ and $[2, 1]^\top/\sqrt{5}$ for scenario 1 and 2, respectively. We can see that the improvement on PEB and MEB with increased speed saturates at around $\unit[0.1]{m/s}$ {(which aligns well with the theoretical analysis as dashed curves)}, whereas the VEB keeps increasing.}\vspace{-5mm}
\label{fig-3}
\end{figure}
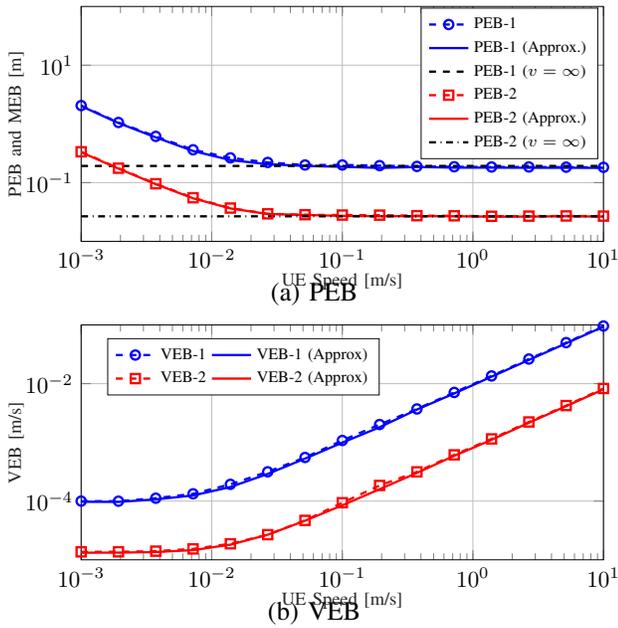


\subsection{Evaluation of the Localization Algorithm}
We evaluate the performance of the proposed localization algorithm with the error bound discussed in Section~\ref{sec:performance_analysis_and_localization_algorithm}. Since channel estimation is not considered in this work, we generate channel parameter vector (by assuming the channel estimation is efficient) following a multi-variable Gaussian distribution as $\tilde \gammav\in \mathcal{N}(\gammav, \mathbf{I}(\gammav)^{-1})$, where $\mathbf{I}(\gammav)$ is the EFIM of the unknown channel parameters obtained from obtained from~\eqref{eq:FIM_measurement}. The results are shown in Fig.~\ref{fig:estiamtor_and_bound} with $500$ simulations performed for each point. We can see that even with a simple algorithm, the localization and mapping results are close to the bound when the transmission power is above $\unit[15]{dBm}$. This is due to the signal from LOS path is much stronger than the NLOS path and hence it is reasonable to search along the LOS direction. Note that the localization and mapping are done with a limited number of measurements within the coherence time, better results can be obtained in tracking scenario with the assist of estimated velocity. 

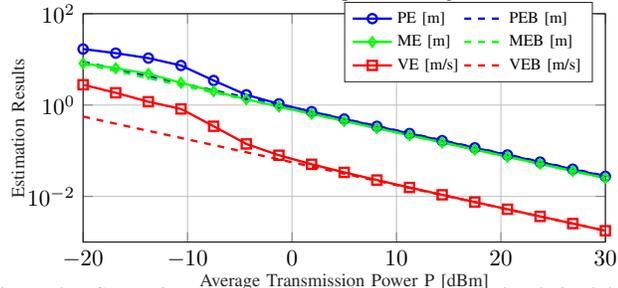
\begin{figure}[htb]
\begin{minipage}[b]{0.98\linewidth}
    \centering
%
%
\begin{tikzpicture}
[scale=1\columnwidth/10cm,font=\normalsize]
\begin{axis}[%
width=8cm,
height=3.5cm,
at={(0, 0)},
scale only axis,
xmin=-20,
xmax=30,
xlabel style={font=\footnotesize\color{white!15!black}, yshift=1 ex},xlabel={UE Speed [m/s]},
xlabel={Average Transmission Power P [dBm]},
ymode=log,
ymin=0.001,
ymax=100,
yminorticks=true,
ylabel style={font=\footnotesize\color{white!15!black}, yshift= -1.5 ex},
ylabel={Estimation Results},
axis background/.style={fill=white},
xmajorgrids,
ymajorgrids,
yminorgrids,
legend style={font=\scriptsize, at={(0.5,0.7)}, anchor=south west, legend cell align=left, align=left, draw=white!15!black, legend columns=2}
]
\addplot [color=blue, line width=1.0pt, mark=o, mark options={solid, blue}]
  table[row sep=crcr]{%
-20	16.853818397864\\
-16.875	13.7047887156354\\
-13.75	10.6729911753109\\
-10.625	7.35941224374096\\
-7.5	3.44970243770965\\
-4.375	1.67656145483772\\
-1.25	1.06167576471805\\
1.875	0.716942967284543\\
5	0.493312178003039\\
8.125	0.342040261929035\\
11.25	0.237967325548852\\
14.375	0.165824028107826\\
17.5	0.115638794328881\\
20.625	0.0806706778397434\\
23.75	0.0562860187121121\\
26.875	0.0392755575611606\\
30	0.0274067848712659\\
};
\addlegendentry{PE [m]}

\addplot [color=blue, dashed, line width=1.0pt]
  table[row sep=crcr]{%
-20	8.84952018232522\\
-16.875	6.1754658445612\\
-13.75	4.30942894209227\\
-10.625	3.00725131907228\\
-7.5	2.09855194680882\\
-4.375	1.46443373240041\\
-1.25	1.02192664796949\\
1.875	0.713131670436434\\
5	0.49764509066275\\
8.125	0.347271964669828\\
11.25	0.242336998210981\\
14.375	0.169110169194753\\
17.5	0.118010248274926\\
20.625	0.0823511605731494\\
23.75	0.0574671585466496\\
26.875	0.0401023408588432\\
30	0.0279846399757762\\
};
\addlegendentry{PEB [m]}

\addplot [color=green, line width=1.0pt, mark=diamond, mark options={solid, green}]
  table[row sep=crcr]{%
-20	8.1620983244932\\
-16.875	6.40175841053234\\
-13.75	4.82909842635068\\
-10.625	3.07144654497101\\
-7.5	2.00662869344338\\
-4.375	1.34355823910902\\
-1.25	0.924103089881719\\
1.875	0.641251841841541\\
5	0.446448279425309\\
8.125	0.311256552211582\\
11.25	0.217132951502559\\
14.375	0.151510086453984\\
17.5	0.105730205005615\\
20.625	0.0737852536199825\\
23.75	0.0514921640432047\\
26.875	0.0359343590082808\\
30	0.0250769431664088\\
};
\addlegendentry{ME [m]}

\addplot [color=green, dashed, line width=1.0pt]
  table[row sep=crcr]{%
-20	7.97884547886667\\
-16.875	5.56788240702403\\
-13.75	3.88543863652454\\
-10.625	2.71137791616304\\
-7.5	1.89208243701218\\
-4.375	1.32035299362329\\
-1.25	0.921382701761617\\
1.875	0.642969029650073\\
5	0.448683454007491\\
8.125	0.313105037126993\\
11.25	0.218494271180899\\
14.375	0.152471985046697\\
17.5	0.106399614499882\\
20.625	0.0742489052153166\\
23.75	0.0518131569516092\\
26.875	0.0361568056189775\\
30	0.0252313248117556\\
};
\addlegendentry{MEB [m]}

\addplot [color=red, line width=1.0pt, mark=square, mark options={solid, red}]
  table[row sep=crcr]{%
-20	2.78501254757689\\
-16.875	1.85857817286902\\
-13.75	1.18318110995233\\
-10.625	0.82231234187681\\
-7.5	0.343803972579983\\
-4.375	0.140814168041175\\
-1.25	0.0795652237833152\\
1.875	0.0505277003616479\\
5	0.0335736253565597\\
8.125	0.0228107455589196\\
11.25	0.0156796042168751\\
14.375	0.0108457988105311\\
17.5	0.00752852285082805\\
20.625	0.00523637209058395\\
23.75	0.00364646441501156\\
26.875	0.00254115790753317\\
30	0.00177168687706555\\
};
\addlegendentry{VE [m/s]}

\addplot [color=red, dashed, line width=1.0pt]
  table[row sep=crcr]{%
-20	0.559639967103051\\
-16.875	0.390533885554478\\
-13.75	0.272526489764078\\
-10.625	0.190177319741873\\
-7.5	0.132711550262553\\
-4.375	0.0926101787373768\\
-1.25	0.0646262151922804\\
1.875	0.0450981495449086\\
5	0.0314708680730213\\
8.125	0.0219613342734441\\
11.25	0.0153252907403406\\
14.375	0.0106944566004793\\
17.5	0.00746291890427096\\
20.625	0.00520785306372916\\
23.75	0.00363419914932636\\
26.875	0.00253605531787165\\
30	0.00176973696570734\\
};
\addlegendentry{VEB [m/s]}

\end{axis}

\end{tikzpicture}%
\end{minipage}
\vspace{-1cm}
\caption{Comparison between simulation results and the derived lower bounds (PEB, VEB, MEB). When $P\ge \unit[5]{dBm}$, the estimation results using the proposed localization and mapping algorithm attach the bounds.}\vspace{-5mm}
\label{fig:estiamtor_and_bound}
\end{figure}

\section{Conclusion}
With a sufficient number of multipaths, UE mobility helps localization and mapping by providing extra channel parameters. We have shown that mobility can enable localization in a SIMO uplink scenario where BS and UE are not synchronized. In addition, even though extra unknowns (i.e., velocity) are introduced, mobility can enhance localization and mapping to some extent with an increased speed. We also analyzed the system performance under different scenarios and evaluated the performance of the proposed localization algorithm. Future works could be the research on channel estimation under UE mobility and Doppler-assisted simultaneous localization and mapping in tracking scenarios.

\section*{Acknowledgment}
This work was supported, in part, by the European Commission through the H2020 project Hexa-X (Grant Agreement no. 101015956), and by the Wallenberg AI, Autonomous Systems and Software Program (WASP) funded by Knut and Alice Wallenberg Foundation.

\balance
\bibliographystyle{IEEEtran}
\bibliography{reference}

\begin{thebibliography}{10}
\providecommand{\url}[1]{#1}
\csname url@samestyle\endcsname
\providecommand{\newblock}{\relax}
\providecommand{\bibinfo}[2]{#2}
\providecommand{\BIBentrySTDinterwordspacing}{\spaceskip=0pt\relax}
\providecommand{\BIBentryALTinterwordstretchfactor}{4}
\providecommand{\BIBentryALTinterwordspacing}{\spaceskip=\fontdimen2\font plus
\BIBentryALTinterwordstretchfactor\fontdimen3\font minus
  \fontdimen4\font\relax}
\providecommand{\BIBforeignlanguage}[2]{{%
\expandafter\ifx\csname l@#1\endcsname\relax
\typeout{** WARNING: IEEEtran.bst: No hyphenation pattern has been}%
\typeout{** loaded for the language `#1'. Using the pattern for}%
\typeout{** the default language instead.}%
\else
\language=\csname l@#1\endcsname
\fi
#2}}
\providecommand{\BIBdecl}{\relax}
\BIBdecl

\bibitem{wymeersch20175g}
H.~Wymeersch, G.~Seco-Granados, G.~Destino, D.~Dardari, and F.~Tufvesson, ``{5G
  mmWave} positioning for vehicular networks,'' \emph{IEEE Wireless Commun.},
  vol.~24, no.~6, pp. 80--86, Dec. 2017.

\bibitem{haddadin2018tactile}
S.~Haddadin, L.~Johannsmeier, and F.~D. Ledezma, ``Tactile robots as a central
  embodiment of the tactile internet,'' \emph{Proc. IEEE}, vol. 107, no.~2, pp.
  471--487, Dec. 2018.

\bibitem{shahmansoori2017position}
A.~Shahmansoori, G.~E. Garcia, G.~Destino, G.~Seco-Granados, and H.~Wymeersch,
  ``Position and orientation estimation through millimeter-wave {MIMO in 5G}
  systems,'' \emph{IEEE Trans. Wireless Commun.}, vol.~17, no.~3, pp.
  1822--1835, Dec. 2017.

\bibitem{jiang2021beamspace}
F.~Jiang, F.~Wen, Y.~Ge, M.~Zhu, H.~Wymeersch, and F.~Tufvesson, ``Beamspace
  multidimensional {ESPRIT} approaches for simultaneous localization and
  communications,'' \emph{arXiv preprint arXiv:2111.07450}, 2021.

\bibitem{chen2021tutorial}
H.~Chen, H.~Sarieddeen, T.~Ballal, H.~Wymeersch, M.-S. Alouini, and T.~Y.
  Al-Naffouri, ``A tutorial on terahertz-band localization for {6G}
  communication systems,'' \emph{Accepted for publication in {IEEE} Commun.
  Surveys Tuts. arXiv preprint arXiv:2110.08581}, 2022.

\bibitem{abu2018error}
Z.~Abu-Shaban, X.~Zhou, T.~Abhayapala, G.~Seco-Granados, and H.~Wymeersch,
  ``Error bounds for uplink and downlink {3D localization in 5G} millimeter
  wave systems,'' \emph{IEEE Trans. Wireless Commun.}, vol.~17, no.~8, pp.
  4939--4954, May. 2018.

\bibitem{fascista2021downlink}
A.~Fascista, A.~Coluccia, H.~Wymeersch, and G.~Seco-Granados, ``Downlink
  single-snapshot localization and mapping with a single-antenna receiver,''
  \emph{IEEE Trans. Wireless Commun.}, vol.~20, no.~7, pp. 4672--4684, Mar.
  2021.

\bibitem{talvitie2019positioning}
J.~Talvitie, T.~Levanen, M.~Koivisto, T.~Ihalainen, K.~Pajukoski, and
  M.~Valkama, ``Positioning and location-aware communications for modern
  railways with {5G} new radio,'' \emph{IEEE Commun. Mag.}, vol.~57, no.~9, pp.
  24--30, Sep. 2019.

\bibitem{win2018theoretical}
M.~Z. Win, Y.~Shen, and W.~Dai, ``A theoretical foundation of network
  localization and navigation,'' \emph{Proc. IEEE}, vol. 106, no.~7, pp.
  1136--1165, Jul. 2018.

\bibitem{kim2022preamble}
Y.~J. Kim, M.~Asim, and Y.~S. Cho, ``Preamble design technique for accurate
  timing/positioning in high {Doppler} environments,'' \emph{IEEE Trans. Veh.
  Technol.}, Mar. 2022.

\bibitem{amar2008localization}
A.~Amar and A.~J. Weiss, ``Localization of narrowband radio emitters based on
  {Doppler} frequency shifts,'' \emph{IEEE Trans. Signal Process.}, vol.~56,
  no.~11, pp. 5500--5508, Aug. 2008.

\bibitem{shames2013doppler}
I.~Shames, A.~N. Bishop, M.~Smith, and B.~D. Anderson, ``Doppler shift target
  localization,'' \emph{IEEE Trans. Aerosp. Electron. Syst.}, vol.~49, no.~1,
  pp. 266--276, Jan. 2013.

\bibitem{han2015performance}
Y.~Han, Y.~Shen, X.-P. Zhang, M.~Z. Win, and H.~Meng, ``Performance limits and
  geometric properties of array localization,'' \emph{IEEE Trans. Inf. Theory},
  vol.~62, no.~2, pp. 1054--1075, Dec. 2015.

\bibitem{kakkavas2019performance}
A.~Kakkavas, M.~H.~C. Garc{\'\i}a, R.~A. Stirling-Gallacher, and J.~A. Nossek,
  ``Performance limits of single-anchor millimeter-wave positioning,''
  \emph{IEEE Trans. Wireless Commun.}, vol.~18, no.~11, pp. 5196--5210, Aug.
  2019.

\bibitem{fascista2019millimeter}
A.~Fascista, A.~Coluccia, H.~Wymeersch, and G.~Seco-Granados, ``Millimeter-wave
  downlink positioning with a single-antenna receiver,'' \emph{IEEE Trans.
  Wireless Commun.}, vol.~18, no.~9, pp. 4479--4490, Jul. 2019.

\bibitem{mendrzik2018harnessing}
R.~Mendrzik, H.~Wymeersch, G.~Bauch, and Z.~Abu-Shaban, ``Harnessing {NLOS}
  components for position and orientation estimation in {5G} millimeter wave
  {MIMO},'' \emph{IEEE Trans. Wireless Commun.}, vol.~18, no.~1, pp. 93--107,
  Oct. 2018.

\bibitem{fascista2022ris}
A.~Fascista, M.~F. Keskin, A.~Coluccia, H.~Wymeersch, and G.~Seco-Granados,
  ``{RIS}-aided joint localization and synchronization with a single-antenna
  receiver: Beamforming design and low-complexity estimation,'' \emph{Accepted
  for publication in {IEEE} J. Sel. Topics Signal Process.. arXiv preprint
  arXiv:2204.13484}, 2022.

\bibitem{miller1981inverse}
K.~S. Miller, ``On the inverse of the sum of matrices,'' \emph{Math. Mag.},
  vol.~54, no.~2, pp. 67--72, Mar. 1981.

\bibitem{nazari2022mmwave}
M.~A. Nazari, G.~Seco-Granados, P.~Johannisson, and H.~Wymeersch, ``{MmWave 6D}
  radio localization with a snapshot observation from a single {BS},''
  \emph{arXiv preprint arXiv:2204.05189}, 2022.

\end{thebibliography}

\end{document}